\documentclass[a4paper,11pt]{article}
\usepackage{jheppub}

\usepackage[utf8]{inputenc}

\usepackage{amsmath}
\usepackage{amssymb}
\usepackage[all]{xy}
\usepackage{tikz}
\usetikzlibrary{calc}

\usepackage[toc,page]{appendix}
\usepackage{enumerate}
\usepackage{tensor}
\usepackage{booktabs}
\usepackage{bbm}
\usepackage{youngtab}
\usepackage{slashed}
\usepackage{multirow}
\usepackage{makecell}
\usepackage{float}
\usepackage{comment}
\usepackage{verbatim}
\usepackage{xcolor}
\usepackage[framemethod=tikz]{mdframed}
\usepackage{soul}
\usepackage{mathtools}
\usepackage[inline]{enumitem}
\usepackage[cbgreek]{textgreek}
\usepackage{accents}

\usepackage{amsthm}
\newtheorem{theorem}{Theorem}
\newtheorem{lemma}{Lemma}
\newtheorem{corollary}{Corollary}[lemma]
\newtheorem{remark}{Remark}[lemma]

\newcommand{\ab}{|}
\newcommand{\der}{\partial}
\newcommand{\de}{\mathrm{d}}

\newcommand{\X}{\mathsf{x}}
\newcommand{\Z}{\mathsf{z}}

\newcommand{\tomega}{\tilde{\omega}}
\newcommand{\tdelta}{\tilde{\delta}}
\newcommand{\tsigma}{\tilde{\sigma}}

\newcommand{\e}{\mathrm{e}}

\usepackage{hyperref}
\hypersetup{
    colorlinks=true,
    linkcolor=green!55!teal,
    citecolor=green!55!teal,
    filecolor=green!55!teal,      
    urlcolor=green!55!teal,
}

\title{Collapsing universe before time}

\author[a]{Gary Shiu,}
\affiliation[a]{Department of Physics, University of Wisconsin-Madison, 1150 University Avenue, Madison, WI 53706, USA}

\author[b]{Flavio Tonioni,}
\affiliation[b]{Instituut voor Theoretische Fysica, KU Leuven, Celestijnenlaan 200D, B-3001 Leuven, Belgium}

\author[c]{Hung V. Tran}
\affiliation[c]{Department of Mathematics, University of Wisconsin-Madison, 480 Lincoln Drive, Madison, WI 53706, USA}

\emailAdd{shiu@physics.wisc.edu}
\emailAdd{flavio.tonioni@kuleuven.be}
\emailAdd{hung@math.wisc.edu}

\abstract{In this note, we prove analytic bounds on the equation of state of a cosmological fluid composed of an arbitrary number of canonical scalars evolving in a negative multi-exponential potential. Because of the negative energy, the universe is contracting and eventually undergoes a big crunch. A contracting universe is a fundamental feature of models of ekpyrosis and cyclic universes, which are a proposed alternative to cosmic inflation to solve the flatness and horizon problems. Our analytic bounds set quantitative constraints on general effective theories of ekpyrosis. When applied to specific top-down constructions, our bounds can be used to determine whether ekpyrosis could in principle be realized. We point out some possible sources of tension in realizing the ekpyrotic universe in controlled constructions of string theory.}

\begin{document}
\maketitle

\section{Introduction}

In physics, one sometimes uncovers new insights by considering the time direction in reverse. For instance, the positron described by negative-energy solutions of the Dirac equation was given the interpretation of an electron moving backward in time. Another famous example is Hawking's singularity theorem: if we reverse the expansion of the universe, the cosmic scale factor approaches zero, leading to a cosmic singularity. If the stress tensor satisfies the dominant energy condition, this cosmic singularity is unavoidable. During cosmic inflation \cite{Starobinsky:1980te,Guth:1980zm,Linde:1981mu,Albrecht:1982wi}, the approximately constant positive potential energy violates the dominant energy condition, and it was originally thought that the initial big-bang singularity could be avoided. However, it has since been shown that inflationary cosmologies are past-incomplete \cite{Borde:2001nh}; see though ref. \cite{Lesnefsky:2022fen}. It remains an open problem to account for the big-bang singularity prior to inflation.

The ekpyrotic universe \cite{Khoury:2001wf,Khoury:2001bz} was proposed as an alternative to inflation. A defining feature of this scenario is the transition from a contracting universe to an expanding phase. While the known laws of physics break down at the cosmic singularity, it has been argued that the cosmological perturbations generated during the contracting phase can account for the observed large-scale structure of the universe. In the present work, we are agnostic about the viability of this scenario. The focus of our study is on the contracting phase far from the cosmic singularity, where the laws of physics as we know them remain valid.

Recently, we derived sharp theoretical bounds on late-time cosmic acceleration \cite{Shiu:2023nph, Shiu:2023fhb}. As we show in the present note, similar arguments can be applied to constrain scenarios involving a collapsing universe. In detail, we find analytic bounds on the equation of state ($w$-parameter) of a cosmological fluid composed of an arbitrary number of canonical scalar fields evolving in a negative multi-exponential potential. Because of the negative energy, the universe
contracts and eventually undergoes a big crunch. The results we derive complement those on cosmic acceleration, which apply to the same class of theories if the potential is positive. In fact, the results of the present note represent a fundamental step in the path towards a full characterization of the model-independent features of scalar-field cosmologies with multi-field multi-exponential potentials, and generalizations thereof. For instance, recently, ref. \cite{VanRiet:2023cca} pointed out that, although negative potential terms reduce deceleration, they do so around solutions that are perturbatively unstable. This is a further obstruction to late-time cosmic acceleration in string compactifications. Our result is going to be an important step in addressing this problem beyond linear stability of exact solutions.

A contracting universe is a fundamental feature of models of ekpyrosis and cyclic universes \cite{Khoury:2001wf, Khoury:2001bz, Steinhardt:2002ih}, where negative exponential potentials are good candidates for modelling steep field-space directions. Just like in refs. \cite{Shiu:2023fhb, Shiu:2023nph}, we exploit the formulation of the cosmological equations in terms of an autonomous system of ordinary differential equations.\footnote{This formulation has recently been widely used; see e.g. refs. \cite{Conlon:2022pnx, Rudelius:2022gbz, Marconnet:2022fmx, Apers:2022cyl, Hebecker:2023qke, VanRiet:2023cca, Andriot:2023wvg, Revello:2023hro}.} Our method allows us to bound $w$ without identifying an actual solution to the field equations, providing an alternative to the analysis of linear perturbation followed e.g. by refs. \cite{Heard:2002dr, Lehners:2007ac, Koyama:2007mg, Tzanni:2014eja}; see also refs. \cite{Giambo:2014jfa, Giambo:2015tja} for non-perturbative results independent of a potential of exponential form.

The paper is organized as follows. In sec. \ref{sec: universal big-crunch bound}, we describe our analytic bounds on $w$. In sec. \ref{sec: phenomenology of ekpyrosis}, we discuss the implications of such results for phenomenological descriptions of ekpyrosis and comment on ekpyrosis in string compactifications with a few examples. In sec. \ref{sec: conclusions}, we recap the main conclusions. In app. \ref{app: bounds on ekpyrotic universes}, we show the mathematical proofs of our bounds, while in apps. \ref{app: negative-potential scaling solutions} and \ref{app: curvature anisotropies} we report useful results to facilitate reading the main text.

\section{Universal big-crunch bound} \label{sec: universal big-crunch bound}
In this section we present the main analytic results. Given a theory of $n$ canonical scalars $\phi^a$, for $a=1,\dots,n$, we consider a non-compact $d$-dimensional spacetime characterized by the usual $d$-dimensional FLRW-metric
\begin{equation} \label{FLRW-metric}
    d \tilde{s}_{1,d-1}^2 = - \de t^2 + a^2(t) \, \de l_{\mathbb{E}^{d-1}}^2,
\end{equation}
with the Hubble parameter $\smash{H = \dot{a}/a}$ and $d>2$. A cosmological fluid is characterized by an equation of state $p = w \rho$. Here, $p$ is the pressure and $\rho$ is the energy density of the fluid, and $w$ is the equation-of-state parameter. In FLRW-cosmology, a canonical scalar-field theory has a $w$-parameter given by $\smash{w = (T - V)/(T + V)}$, where $T$ and $V$ are the kinetic and potential energy densities, respectively.

Here, we find universal bounds on the $w$-parameter for theories with arbitrary multi-field multi-exponential potentials $V$ under the assumption that their sum is negative, i.e. $V<0$. The cosmological equations indicate that the parameter $\smash{\epsilon = -\dot{H}/H^2 = 1-a \ddot{a}/\dot{a}^2}$ is related to the parameter $w$ as
\begin{equation}
    w = -1 + \dfrac{2 \epsilon}{d-1}.
\end{equation}
Therefore, the condition that $w \gg 1$, which we will look for below, corresponds to the condition $- a \ddot{a}/\dot{a}^2 \gg d-2$; in particular, it implies that $\ddot{a} < 0$, which in case of a contracting universe, i.e. with $\dot{a}<0$, means that the scale factor is decreasing in an accelerated way, namely it decreases at a faster rate as time passes.

\subsection{Negative-definite potentials}

To start, we consider a negative-definite multi-field multi-exponential potential
\begin{equation} \label{generic negative-definite exponential potential}
    V = - \sum_{i = 1}^m K_i \, \e^{- \kappa_d \gamma_{i a} \phi^a},
\end{equation}
with $K_i > 0$ for all terms  $i=1,\dots,m$. Let $t_0$ be an initial time such that the Hubble parameter is negative, i.e. $H(t_0) < 0$. Then, there exists a finite time $t_\star$, such that $\smash{t_0 < t_\star \leq t_0 - 1/[(d-1) H(t_0)]}$, at which the Hubble parameter diverges, i.e. $H(t_\star) = - \infty$. A proof is in appendix \ref{app: bounds on ekpyrotic universes}; see corollary \ref{corollary: g-divergence (E)}.

For each field $\phi^a$, let $\smash{\gamma^a = \min_i {\gamma_{i}}^a}$ and $\smash{\Gamma^a = \max_i {\gamma_{i}}^a}$; then, let $\smash{\Gamma_\star^a = \max \lbrace \ab \gamma^a \ab, \ab \Gamma^a \ab \rbrace}$. If $\smash{(\Gamma_\star)^2 \geq \Gamma_d^2 = 4 \, (d-1)/(d-2)}$, then there exists a time $\smash{\bar{t}}$ in the interval $\smash{t_0 \leq \bar{t} < t_\star}$, such that, at all times $t$ in the interval $\smash{\bar{t} \leq t < t_\star}$, the $w$-parameter of the theory is bounded from above as
\begin{equation} \label{upper w-bound for negative-definite potentials}
    w \leq -1 + \dfrac{1}{2} \, \dfrac{d-2}{d-1} \, (\Gamma_\star)^2.
\end{equation}
If $\smash{(\Gamma_\star)^2 < \Gamma_d^2}$, then sufficiently close to the time $\smash{t=t_\star}$ the $w$-parameter is arbitrarily close to $w=1$. This is because the potential is sufficiently shallow that eventually, in a rapidly contracting universe, the potential energy becomes negligible compared to the kinetic energy.\footnote{This is opposite for positive potentials \cite{Shiu:2023fhb}, for which the universe is expanding and it is a steep potential that leads to the kinetic energy overshadowing all the potential energy.} A rigorous mathematical proof of the bound in eq. (\ref{upper w-bound for negative-definite potentials}) is in appendix \ref{app: bounds on ekpyrotic universes}; see corollary \ref{corollary: (x)^2-bounds (E)}. To identify the optimal version of the bound, we can work in a basis such that $\Gamma^a \geq 0$ for all fields $\phi^a$. If $\ab \Gamma^a \ab \geq \ab \gamma^a \ab$, this is automatic; else, if $\ab \Gamma^a \ab < \ab \gamma^a \ab$, one can just flip the field sign and work in the new basis. In the end, one has $\smash{\Gamma_\star^a = \Gamma^a \geq 0}$. The optimal bound is the one where the length of the vector $\Gamma_\star$ is minimal, which can be achieved by an $\mathrm{O}(n)$-rotation in field space. Therefore, the optimal version of the bound can be expressed as
\begin{equation} \label{optimal upper w-bound for negative-definite potentials}
    w \leq -1 + \dfrac{1}{2} \, \dfrac{d-2}{d-1} \, (\tilde{\Gamma}_\star)^2,
\end{equation}
in a basis $\smash{(\tilde{\phi}^a)_{a=1}^n}$, where we defined
\begin{equation}
    (\tilde{\Gamma}_\star)^2 = \min_{\mathrm{R} \in \mathrm{O}(n)} [\Gamma_\star(\mathrm{R})]^2.
\end{equation}
Here, $\smash{\mathrm{R} \in \mathrm{O}(n)}$ indicates all possible $\mathrm{O}(n)$-rotations in the $n$-dimensional field-space basis and $\smash{[\Gamma_\star(\mathrm{R})]^2}$ represents the $\smash{(\Gamma_\star)^2}$-coefficient computed in the $\smash{\mathrm{R}^{-1}}$-rotated field-space basis.

It turns out that the $w$-parameter is also bounded from below. For each field $\phi^a$, if $\gamma^a > 0$, let $\smash{\gamma_\diamond^a = \gamma^a}$; if $\gamma^a \leq 0$, then let $\smash{\gamma_\diamond^a = 0}$. Then, if $\smash{(\gamma_\diamond)^2 \geq \Gamma_d^2}$, the $w$-parameter is also bounded from below as
\begin{equation} \label{lower w-bound for negative-definite potentials}
    w \geq -1 + \dfrac{1}{2} \, \dfrac{d-2}{d-1} \, (\gamma_\diamond)^2.
\end{equation}
If $\smash{(\gamma_\diamond)^2 < \Gamma_d^2}$, then the bound reads trivially because one always has $w \geq 1$. Rigorous mathematical proofs of the bound in eq. (\ref{lower w-bound for negative-definite potentials}) are in appendix \ref{app: bounds on ekpyrotic universes}; see corollary \ref{corollary: (x)^2 lower bound (E)}. Formally, this is clearly analogous to the late-time bound on cosmic acceleration for positive-definite potentials discussed in ref. \cite{Shiu:2023fhb}, and in this case, too, there exists a field-space basis in which the lower bound is optimized, i.e. pushed to the highest value as possible. To find the optimal value of $\smash{(\gamma_\diamond)^2}$ one must find the distance of the origin from the coupling convex hull. We may write
\begin{equation} \label{optimal lower w-bound for negative-definite potentials}
    w \geq -1 + \dfrac{1}{2} \, \dfrac{d-2}{d-1} \, (\hat{\gamma}_\diamond)^2,
\end{equation}
in a basis $\smash{(\hat{\phi}^a)_{a=1}^n}$, where
\begin{equation}
    (\hat{\gamma}_\diamond)^2 = \max_{\mathrm{R} \in \mathrm{O}(n)} [\gamma_\diamond(\mathrm{R})]^2.
\end{equation}
Here we used a different notation for the basis since, in general, the field bases where the upper bound in eq. (\ref{upper w-bound for negative-definite potentials}) and the lower bound in eq. (\ref{lower w-bound for negative-definite potentials}) are optimal do not coincide.

It is convenient to represent the couplings in a vector space, in terms of vectors $\mu_i$ with components $\smash{(\mu_i)_a = \gamma_{ia}}$, like in ref. \cite{Shiu:2023fhb}. A graphical interpretation of the bounds in eqs. (\ref{upper w-bound for negative-definite potentials}, \ref{lower w-bound for negative-definite potentials}) and of their optimal form in eq. (\ref{optimal upper w-bound for negative-definite potentials}, \ref{optimal lower w-bound for negative-definite potentials}), is in figs. \ref{fig.: optimal ekpyrosis bound 1} and \ref{fig.: optimal ekpyrosis bound 2}.

\begin{figure}[ht]
    \centering
    \begin{tikzpicture}[xscale=1.05,yscale=1.05,every node/.style={font=\normalsize}]
    
    \node[align=left] at (6.6,3.8){$\phi^a = {\color{cyan} \phi^1}, {\color{magenta} \phi^2}$ \\[1.0ex] $\gamma_{ia} = \left(\!\begin{array}{cc}
    \gamma_{1 \color{cyan} 1} & \gamma_{1 \color{magenta} 2} \\
    \gamma_{2 \color{cyan} 1} & \gamma_{2 \color{magenta} 2} \\
    \gamma_{3 \color{cyan} 1} & \gamma_{3 \color{magenta} 2}
    \end{array}\!\right)$};

    \draw[densely dotted,green!80!orange,fill=green!10!white] (-0.5,3.5) -- (-1.9,4.2) -- (-2.5,3) -- (2.9,0.3) -- (3.5,1.5) -- (-0.5,3.5);
    \draw[densely dotted,green!80!orange] (-2.5,3) -- (-3.2,1.6);

    \draw[orange,fill=orange!15!white] (3.5,1.5) -- (-2.5,3) -- (-0.5,3.5) -- (3.5,1.5);

    \draw[->, thin, green] (0.8,0) arc (0:atan(2):0.8);
    \draw[->, thin, green] (0,0.8) arc (90:90+atan(2):0.8);

    \draw[->, thin, green!45!olive!55!white] (0,0) -- (3.5,3.5) node[above right,black]{};    
    \draw[dotted, thin, green!45!olive!55!white] (3.5,0) -- (3.5,3.5) node[above right,black]{};
    \draw[dotted, thin, green!45!olive!55!white] (-0.5,3.5) -- (3.5,3.5) node[above right,black]{};
    
    \draw[->, ultra thin, gray!60!white] (-3,0) -- (4.5,0) node[below]{$\gamma_{\star 1}$};
    \draw[->, ultra thin, gray!60!white] (0,-0.8) -- (0,5.0) node[left]{$\gamma_{\star 2}$};
    
    \draw[->, thick, teal] (0,0) -- (3.5,1.5) node[right,black]{$\mu_1$};
    \draw[densely dotted, ultra thin, gray!60!white] (3.5,1.5) -- (3.5,0) node[below]{$\gamma_{1 \color{cyan!30!white} 1}$};
    \draw[densely dotted, ultra thin, gray!60!white] (3.5,1.5) -- (0,1.5) node[left]{$\gamma_{1 \color{magenta!30!white} 2}$};

    \draw[->, thick, teal] (0,0) -- (-2.5,3) node[left,black]{$\mu_2$};
    \draw[densely dotted, ultra thin, gray!60!white] (-2.5,3) -- (-2.5,0) node[below]{$\gamma_{2 \color{cyan!30!white} 1}$};
    \draw[densely dotted, ultra thin, gray!60!white] (-2.5,3) -- (0,3) node[left]{$\gamma_{2 \color{magenta!30!white} 2}$};

    \draw[->, thick, teal] (0,0) -- (-0.5,3.5) node[left,black]{$\mu_3$};
    \draw[densely dotted, ultra thin, gray!60!white] (-0.5,3.5) -- (-0.5,0) node[below]{$\gamma_{3 \color{cyan!30!white} 1}$};
    \draw[densely dotted, ultra thin, gray!60!white] (-0.5,3.5) -- (0,3.5) node[above left]{$\gamma_{3 \color{magenta!30!white} 2}$};

    \draw[->,rotate=atan(2), ultra thin] (-0.8,0) -- (3.5,0) node[right]{$\tilde{\gamma}_{\star 1}$};
    \draw[->,rotate=atan(2), ultra thin] (0,-0.8) -- (0,4) node[below]{$\tilde{\gamma}_{\star 2}$};

    \draw[->, thick, green!65!olive] (0,0) -- (-1.9,4.2) node[above]{$(\tilde{\Gamma}_\star)^2$};

    \draw[->, thick, purple] (0,0) -- (19/34,38/17) node[above]{$(\hat{\gamma}_\diamond)^2$};
    
    \end{tikzpicture}
    \caption{A representation of the optimal $w$-parameter bound $\smash{w \leq -1 + [(d-2)/(d-1)] \, (\tilde{\Gamma}_\star)^2/2}$: lighter lines denote the original field basis, with a non-optimal upper bound, while darker lines denote the basis with the  minimal upper bound; the lower bound $\smash{(\hat{\gamma}_\diamond)^2}$ is also displayed.}
    \label{fig.: optimal ekpyrosis bound 1}
\end{figure}
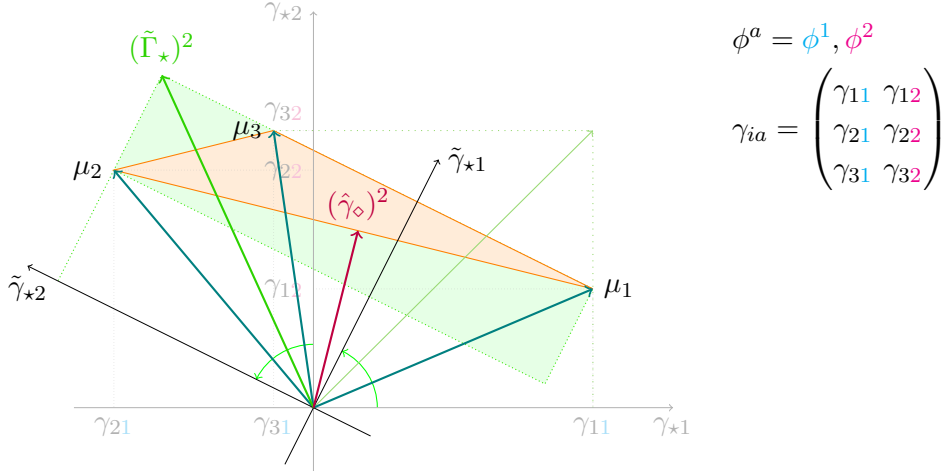

As shown in fig. \ref{fig.: optimal ekpyrosis bound 1}, one has to build the rectangular box with the smallest area that contains all the coupling vectors $\mu_i$. This box has $2^n$ vertices and the largest squared distance between the origin and such vertices provides the smallest -- i.e. the optimal -- value for the quantity $\smash{(\tilde{\Gamma}_\star)^2}$. In general, the problem must be treated as a multi-field one as, effectively, the optimal bound is such that the $w$-parameter receives contributions from each field direction\footnote{This is opposite to the late-time behavior of scalar cosmologies with positive-definite potentials \cite{Shiu:2023nph, Shiu:2023fhb}, where one needs to compute the minimal distance of the origin from the coupling convex hull: such a distance is related to the fact that the problem is asymptotically one-dimensional in field space.} in the way that makes all directions as shallow as possible, compatibly with a field-basis rotation.
The fact that negative-definite potentials, contrarily to positive-definite ones, do not allow one to easily single out a unique field direction that is relevant asymptotically can be understood due to a high sensitivity to the initial conditions. Positive-definite potentials have attractors that correspond to valleys of the potential \cite{Shiu:2023fhb}, so an overall negative sign turns the valley into a crest separating steep cliffs on its sides. These instabilities have also been discussed by refs. \cite{Levy:2015awa, Ijjas:2021zyf}. Without specifying the initial conditions, which our analysis is independent of, the bounds in eqs. (\ref{optimal upper w-bound for negative-definite potentials}, \ref{optimal lower w-bound for negative-definite potentials}) are not necessarily saturated by an actual time-dependent solution. In other words, depending on the initial conditions, a solution may either saturate our bounds or satisfy a strict inequalities. It is also possible to provide a physical interpretation of the upper bound of eq. (\ref{optimal upper w-bound for negative-definite potentials}) in the optimal basis, since in the latter we can express the potential as
\begin{equation}
    V = V_\star + \overline{V},
\end{equation}
where
\begin{equation} \label{optimal-bound basis potential}
        V_\star = - \sum_{\underaccent{\dot}{a}=1}^n \e^{- \kappa_d \tilde{\Gamma}_{\star \underaccent{\dot}{a}} \tilde{\phi}^{\underaccent{\dot}{a}}} \biggl[ \sum_{i_{\underaccent{\dot}{a}}} K_{i_{\underaccent{\dot}{a}}} \, \e^{- \kappa_d \sum_{\underaccent{\dot}{b} \neq \underaccent{\dot}{a}} \tilde{\gamma}_{i_{\underaccent{\dot}{a}} {\underaccent{\dot}{b}}} \tilde{\phi}^{\underaccent{\dot}{b}}} \biggr].
\end{equation}
Here, we have split the potential into two parts: the part $V_\star$ contains, for at least one field, the terms with the largest-possible coupling available in the potential; the part $\overline{V}$ is such that all fields therein appear with couplings that, in absolute value, are smaller than the largest available coupling, which is in one of the terms making up $V_\star$. In particular, in $V_\star$, we have split the sum over the potential terms into $n$ sums of subsets of terms labelled by an index $\smash{i_{\underaccent{\dot}{a}}}$, and the $\cdot$-underscript means that there is no Einstein summation. In each term, the exponential coupling $\smash{\tilde{\Gamma}_{\star \underaccent{\dot}{a}}}$ for the field $\smash{\tilde{\phi}^{\underaccent{\dot}{a}}}$
is larger than those appearing in other potential terms for the same field, in absolute value. Therefore, $\smash{\tilde{\Gamma}_{\star \underaccent{\dot}{a}}}$ determines the maximum steepness that the field $\smash{\tilde{\phi}^{\underaccent{\dot}{a}}}$ can experience, which intuitively provides the upper bound for the $w$-parameter in eq. (\ref{optimal upper w-bound for negative-definite potentials}). On the other hand, in the basis that provides the optimal lower bound of eq. (\ref{optimal lower w-bound for negative-definite potentials}), it is easy to see that $\smash{\hat{\gamma}_\diamond}$ provides the minimum steepness that at least one field universally experiences from all potential terms. The remaining fields can possibly increase further the steepness, hence eq. (\ref{optimal lower w-bound for negative-definite potentials}) is a lower bound.\footnote{This is analogous to the interpretation of the lower bound on late-time acceleration for positive-definite potentials in refs. \cite{Shiu:2023nph, Shiu:2023fhb}.}

\begin{figure}[ht]
    \centering
    \begin{tikzpicture}[xscale=0.92,yscale=0.92,every node/.style={font=\normalsize}]
    
    \node[align=left] at (8.0,3.8){$\phi^a = {\color{cyan} \phi^1}, {\color{magenta} \phi^2}$ \\[1.0ex] $\gamma_{ia} = \left(\!\begin{array}{cc}
    \gamma_{1 \color{cyan} 1} & \gamma_{1 \color{magenta} 2} \\
    \gamma_{2 \color{cyan} 1} & \gamma_{2 \color{magenta} 2} \\
    \gamma_{3 \color{cyan} 1} & \gamma_{3 \color{magenta} 2}
    \end{array}\!\right)$};

    \begin{scope}
        \clip (-2.6,-1.8) rectangle (5,5);
        \fill[cyan!5!white] (5,2) -- (-2,5) -- (-5,-2) -- (2,-5) -- (5,2);
    \end{scope}

    \draw[densely dotted,green!80!orange,fill=green!10!white] (5,2) -- (-1,4) -- (-1.95,1.15) -- (4.05,-0.85) -- (5,2);    

    \begin{scope}
        \clip (-2.6,-1.8) rectangle (5,5);
        \draw[dotted, cyan] (5,2) -- (-2,5) -- (-5,-2) -- (2,-5) -- (5,2);
    \end{scope}

    \draw[orange,fill=orange!15!white] (5,2) -- (-1,4) -- (-1.5,1) -- (5,2);

    \draw[->, thin, green] (0.8,0) arc (0:-atan(10/31):0.8);
    \draw[->, thin, green] (0,0.8) arc (90:90-atan(10/31):0.8);
    
    \draw[->, ultra thin, gray!60!white] (-1.5,0) -- (5.5,0) node[right]{$\gamma_{\star 1}$};
    \draw[->, ultra thin, gray!60!white] (0,-0.8) -- (0,5.0) node[left]{$\gamma_{\star 2}$};
    
    \draw[->, thick, teal] (0,0) -- (5,2) node[right,black]{$\mu_1$};
    \draw[densely dotted, ultra thin, gray!60!white] (5,2) -- (5,0) node[below]{$\gamma_{1 \color{cyan!30!white} 1}$};
    \draw[densely dotted, ultra thin, gray!60!white] (5,2) -- (0,2) node[left]{$\gamma_{1 \color{magenta!30!white} 2}$};

    \draw[->, thick, teal] (0,0) -- (-1,4) node[left,black]{$\mu_2$};
    \draw[densely dotted, ultra thin, gray!60!white] (-1,4) -- (-1,0) node[below]{$\gamma_{2 \color{cyan!30!white} 1}$};
    \draw[densely dotted, ultra thin, gray!60!white] (-1,4) -- (0,4) node[left]{$\gamma_{2 \color{magenta!30!white} 2}$};

    \draw[->, thick, teal] (0,0) -- (-1.5,1) node[below left,black]{$\mu_3$};
    \draw[densely dotted, ultra thin, gray!60!white] (-1.5,1) -- (-1.5,0) node[below]{$\gamma_{3 \color{cyan!30!white} 1}$};
    \draw[densely dotted, ultra thin, gray!60!white] (-1.5,1) -- (0,1) node[below left]{$\gamma_{3 \color{magenta!30!white} 2}$};

    \draw[->, thin, green!45!olive!55!white] (0,0) -- (5,4) node[above right,black]{};    
    \draw[dotted, thin, green!45!olive!55!white] (5,0) -- (5,4) node[above right,black]{};
    \draw[dotted, thin, green!45!olive!55!white] (-1,4) -- (5,4) node[above right,black]{};

    \draw[->,rotate=-atan(10/31), ultra thin] (-0.8,0) -- (5.0,0) node[right]{$\tilde{\gamma}_{\star 1}$};
    \draw[->,rotate=-atan(10/31), ultra thin] (0,-0.8) -- (0,4.6) node[left]{$\tilde{\gamma}_{\star 2}$};

    \draw[->, thick, green!65!olive] (0,0) -- (5,2) node[below,pos=0.65]{$(\tilde{\Gamma}_\star)^2$};

    \draw[->, thick, purple] (0,0) -- (-32/173, 208/173) node[right]{$(\hat{\gamma}_\diamond)^2$};

    \begin{scope}[xshift=265pt,scale=0.4]
    
        \fill[cyan!5!white] (5,2) -- (-2,5) -- (-5,-2) -- (2,-5) -- (5,2);
        
        \draw[densely dotted,green!80!orange,fill=green!10!white] (5,2) -- (-1,4) -- (-1.95,1.15) -- (4.05,-0.85) -- (5,2);
        
        \draw[dotted, cyan] (5,2) -- (-2,5) -- (-5,-2) -- (2,-5) -- (5,2);

        \draw[orange,fill=orange!15!white] (5,2) -- (-1,4) -- (-1.5,1) -- (5,2);
    
        \draw[->, thick, teal] (0,0) -- (5,2) node[right,black]{$\mu_1$};

        \draw[->, thick, teal] (0,0) -- (-1,4) node[left,black]{$\mu_2$};

        \draw[->, thick, teal] (0,0) -- (-1.5,1) node[below left,black]{$\mu_3$};

        \draw[->, thick, green!65!olive] (0,0) -- (5,2) node[below,pos=0.65]{$(\tilde{\Gamma}_\star)^2$};

        \draw[->, thick, purple] (0,0) -- (-32/173, 208/173) node[above right]{$(\hat{\gamma}_\diamond)^2$};
    \end{scope}
    
    \end{tikzpicture}
    \caption{A representation of the optimal $w$-parameter bound $\smash{w \leq -1 + [(d-2)/(d-1)] \, (\tilde{\Gamma}_\star)^2/2}$: lighter lines denote the original field basis, with a non-optimal upper bound, while darker lines denote the basis with the minimal upper bound; the lower bound $\smash{(\hat{\gamma}_\diamond)^2}$ is also displayed. The largest square box with a coupling as a vertex is shown, too, with a rescaled version on the side.}
    \label{fig.: optimal ekpyrosis bound 2}
\end{figure}
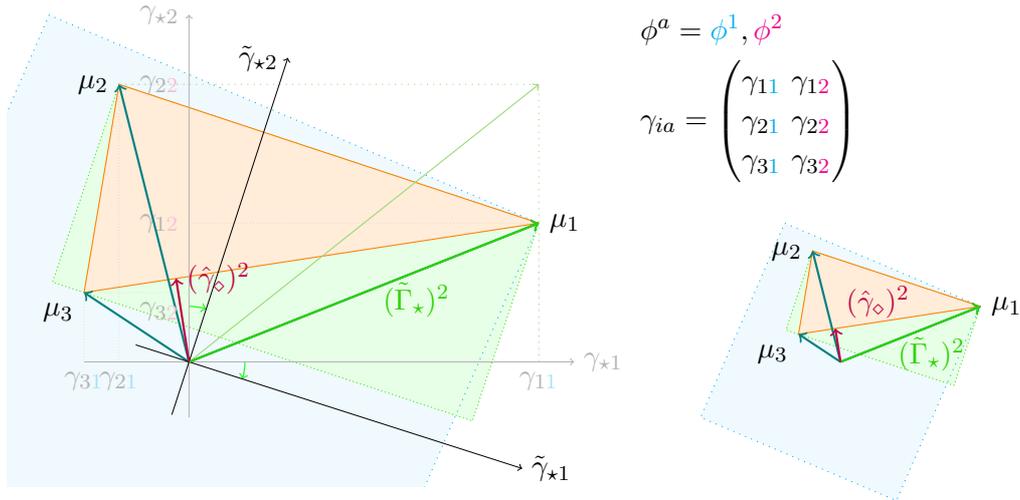

As shown in fig. \ref{fig.: optimal ekpyrosis bound 2}, there exist cases where the minimal – i.e. optimal – bound on $w$ can be computed simply as the longest coupling vector, and therefore one can consider the potentials one by one. Geometrically, this is the case if in the coupling space there is a square box symmetric about the origin that contains all the coupling vectors and such that one vertex is exactly a coupling vector.
In such cases, the optimal bound is the supremum of the distance of the origin from the points belonging to the convex hull of the exponential couplings. For example, one of the cases in which the optimal bound corresponds to the longest coupling vector is for a diagonal matrix $\smash{\gamma_{ia} = \gamma_{a} \delta_{(ia)}}$. Such potentials are studied in ref. \cite{Koyama:2007mg}, which shows that there are perturbatively-stable solutions corresponding to the single-field evolution $\smash{\phi^a(t) = \phi_0 + 2/(\kappa_d \gamma_{a}) \, \mathrm{ln} \, (t/t_0)}$; the solution which is an attractor depends on the initial conditions. In view of our bound in eq. (\ref{optimal upper w-bound for negative-definite potentials}), we conclude that $\smash{w \leq - 1 + (d-2)/[2(d-1)] \, \bigl( \max_{a} \ab \gamma_{a} \ab \bigr)^2}$, which is indeed the maximum value for $w$ that one can find according to ref. \cite{Koyama:2007mg} with the appropriate initial conditions. Our results do not involve the choice of initial conditions and apply to any negative-definite multi-field multi-exponential potential.

\subsection{Arbitrary potentials}

To generalize the results above, we consider a generic multi-field multi-exponential potential
\begin{equation} \label{generic exponential potential}
    V = \sum_{i_+ = 1}^{m_+} \Lambda_{i_+} \, \e^{- \kappa_d \gamma_{i_+ a} \phi^a} - \sum_{i_- = 1}^{m_-} K_{i_-} \, \e^{- \kappa_d \gamma_{i_- a} \phi^a},
\end{equation}
with $\smash{\Lambda_{i_+} > 0}$ and $\smash{K_{i_-} > 0}$ for all terms $\smash{i_\pm=1,\dots,m_\pm}$. Let $t_0$ be an initial time such that the Hubble parameter $\smash{H=\dot{a}/a}$ is negative, i.e. $H(t_0) < 0$, and let the initial conditions be such that $V<0$ at all times. Then, there exists a finite time $t_\star$, such that $\smash{t_0 < t_\star \leq t_0 - 1/[(d-1) H(t_0)]}$, at which the Hubble parameter diverges, i.e. $H(t_\star) = - \infty$. A proof is in appendix \ref{app: bounds on ekpyrotic universes}; see corollary \ref{corollary: g-divergence (E)}.

In this case, we are only able to prove a lower bound on the $w$-parameter. For all fields $\smash{\phi^a}$, let $\smash{\gamma_-^a \geq \Gamma_+^a}$: if $\smash{\gamma_-^a > 0}$, let $\smash{\gamma_\diamond^a = \gamma_-^a}$; else, let $\smash{\gamma_\diamond^a = 0}$. Then, if $\smash{(\gamma_\diamond)^2 \geq \Gamma_d^2}$, close to the big crunch the $w$-parameter is bounded from below as
\begin{equation} \label{lower w-bound}
    w \geq - 1 + \dfrac{1}{2} \, \dfrac{d-2}{d-1} \, (\gamma_\diamond)^2.
\end{equation}
If $\smash{\gamma_+^a \geq \Gamma_-^a}$ for a field, then one can redefine this field as $\smash{{\phi'}{}^a = -\phi^a}$ and find the same bound in terms of the flipped $\gamma$-coefficients. If for a field none of these orderings is in place, then we must simply set $\smash{\gamma_\diamond^a = 0}$ for that field. Again, if $\smash{(\gamma_\diamond)^2 < \Gamma_d^2}$ the bound reads trivially because one always has $w \geq 1$. Rigorous mathematical proofs of the bound in eq. (\ref{lower w-bound}) are in appendix \ref{app: bounds on ekpyrotic universes}; see corollary \ref{corollary: (x)^2 lower bound 2 (E)}.

\section{Phenomenology of ekpyrosis} \label{sec: phenomenology of ekpyrosis}
While our goal is not to endorse theories of the history of the universe which involve ekpyrosis, our bounds find a natural application in constraining models of ekpyrosis. Here we discuss the implications of our bounds on $w$ for phenomenological descriptions of ekpyrosis and in the search of ekpyrosis in string constructions.

\subsection{Phenomenological descriptions of ekpyrosis}
An ekpyrotic phase of the history of the universe \cite{Khoury:2001wf, Khoury:2001bz, Steinhardt:2002ih} is a proposed alternative to the theory of cosmic inflation in which a long-lived contracting phase solves the flatness problem; for a review, see e.g. refs. \cite{Quevedo:2002xw, Lehners:2008vx, Lehners:2010fy, Cicoli:2023opf, Brandenberger:2023ver}. In a contracting phase, a fluid with $w \gg 1$ would wipe out all other contributions to the energy density before a big crunch, which would then followed by the standard thermal history of the universe. A set of perfect fluids with energy densities $\rho_\alpha$ and equation-of-state parameter $w_\alpha$ gives the Friedmann equation
\begin{align*}
    H^2 = \dfrac{2 \kappa_d^2}{(d-1)(d-2)} \sum_\alpha \rho_{\alpha,0} \, \biggl( \dfrac{a_0}{a} \biggr)^{\!\! (d-1) (1 + w_\alpha)} - \dfrac{k}{a^2},
\end{align*}
where $k$ represents the possible external curvature of the FLRW-metric. In an expanding universe, an early phase of cosmic inflation wipes out the external curvature because all contributions to the energy density of the universe fall off over time with respect to the contribution of the nearly-constant inflaton potential energy ($w=-1$). In an ekpyrotic phase of a contracting universe, the solution is opposite and all energy-density contributions blow up, but less quickly than the contribution of a scalar field that is falling down a steep negative potential, with the maximal $w$ of all.\footnote{A non-flat contracting universe with anisotropies and a inhomogeneous curvature converges to a homogeneous, flat and isotropic universe if it contains energy with equation of state $w>1$: this is in essence the cosmic no-hair theorem for ekpyrosis \cite{Erickson:2003zm}.} A contracting universe also puts in causal contact points in large regions, thus solving the horizon problem. Therefore, ekpyrosis was claimed to prepare the universe in similar conditions as those after cosmic inflation.\footnote{In ref. \cite{Khoury:2001wf}, ekpyrosis is modelled through a bulk brane moving from a boundary brane to another; in ref. \cite{Khoury:2001bz}, the boundary branes themselves move, thus the embedding space contracts to zero size and provides a singularity at the moment of collision; ref. \cite{Steinhardt:2002ih} introduced the idea of a cyclic cosmology (see also e.g. ref. \cite{Ijjas:2019pyf} for recent advancements). Contracting phases of the history of the universe also appear in bouncing cosmologies; for a review, see e.g. ref. \cite{Brandenberger:2016vhg}.}

A natural candidate to realize a phase of ekpyrosis is a scalar-field theory with negative energy density.\footnote{A theory with just negative-definite multi-exponential potentials would not be a complete implementation of the ekpyrotic scenario, since the potential would end up at negative infinity. Nonetheless, for the phase when the potential falls steeply, it is legitimate to approximate the potential in terms of negative-definite terms, for which our results apply. This is the same kind of effective descriptions adopted in the literature, which discusses scaling solutions.} For instance, a simple effective description of ekpyrosis is through a single scalar field with a negative-definite exponential potential
\begin{align*}
    V = - K \, \e^{- \kappa_d \gamma \varphi},
\end{align*}
if $\smash{\gamma > 2 \sqrt{d-1}/\sqrt{d-2}}$. This potential is supposed to only model the phase of ekpyrosis, not the full potential. An exact solution reads $\smash{\varphi(t) = \varphi(t_0) + [2/ (\kappa_d \gamma)] \, \mathrm{ln} (t/t_0)}$ and $H(t) = [4/(d-2) \gamma^2] \, 1/t$, with $t < 0$. The ekpyrotic phase does not last right up to the big crunch; rather, the potential flattens and thus becomes irrelevant as the big crunch is near. Then, the attractor solution reads $\smash{\varphi(t') = \varphi(t_0') + [\sqrt{d-2}/ (\kappa_d \sqrt{d-1})] \, \mathrm{ln} (t'/t'_0)}$ and $H(t') = [1/(d-1)] \, 1/t'$, with $t' < 0$ a new shifted time variable. Note that this formally means that $V/H^2=0$. In a phase of kination with negative Hubble parameter the scalar fluid also has $w=1$, so time evolution does not bring anisotropies back in. In a consistent model of ekpyrosis, the other contributions have been washed away during the ekpyrotic phase to such an extent that they do not have time to come to dominate before right around the big crunch time \cite{Erickson:2003zm}. The big crunch in the $d$-dimensional theory is not necessarily a singularity of the higher-dimensional theory the former is embedded in; for instance, an appropriate time dependence of the scalar fields may get rid of the singularity as the scale factor approaches zero. The subsequent transition towards the standard expanding phase of the universe requires a fluid violating the null-energy condition. Such a transition is investigated in the so-called new ekpyrotic models \cite{Buchbinder:2007ad, Buchbinder:2007tw, Buchbinder:2007at}. For instance, in string-theoretic scenarios, an S-brane was argued to allow for a transition into the standard big-bang scenario \cite{Brandenberger:2020tcr, Brandenberger:2020eyf, Brandenberger:2020wha}; other proposals are in refs. \cite{Easson:2011zy, Ijjas:2016pad, Ijjas:2016tpn, Ijjas:2016vtq, Dobre:2017pnt}.

In our analysis, we are discussing the ekpyrotic phase, i.e. the phase where the potential
is negative and steep. Our optimal bounds in eqs. (\ref{optimal upper w-bound for negative-definite potentials}, \ref{optimal lower w-bound for negative-definite potentials}) provide us with an analytic measure of the couplings that would be needed in order to ensure one of the conditions for ekpyrosis, i.e. a large-enough $w$-parameter before the potential flattens and the big crunch takes place. To determine the value of $w$ that could provide a phase of ekpyrosis, one needs to compare with the other otherwise-dominating terms. For instance, matter, radiation and curvature anisotropies have $w=0, 1/(d-1),1$, respectively (see appendix \ref{app: curvature anisotropies} for the calculation of $w$ for curvature anisotropies). Therefore, a minimal requirement for ekpyrosis is to have $w > 1$. However, to make sure that all possible contributions to the Hubble parameter are sufficiently unimportant, one would have to look for $w \gg 1$. This requirement implies that, to make this possible at late-enough times, the vector ${\tilde{\Gamma}_\star^a}$ should have a length
\begin{equation} \label{ekpyrosis condition}
    (\tilde{\Gamma}_\star)^2 \gg \Gamma_d^2 = 4 \, \dfrac{(d-1)}{(d-2)}
\end{equation}
in order not to be a priori incompatible with ekpyrosis, even before considering all the other challenges related to the big crunch. This is a requirement, due to the bound in eq. (\ref{optimal upper w-bound for negative-definite potentials}), that can be tested in phenomenological descriptions of ekpyrosis. Similarly, the bound in eq. (\ref{optimal lower w-bound for negative-definite potentials}) can be used to single out potentials for which an ekpryotic phase is inevitable.

In the literature, the conditions for an ekpyrotic phase are often formulated in terms of the so-called fast-roll parameters \cite{Gratton:2003pe, Khoury:2003rt}, defined as $\smash{\tilde{\epsilon}_V = [4/(d-2)] \, \kappa_d^2 V^2/\der_a V \der^a V}$ and $\smash{\ab \tilde{\eta}_V \ab = [1/(d-2)] \, \kappa_d^2 V^2 / \der_a \der^a V}$, where $\der_a = \der / \der \phi^a$. Requirements for ekpyrosis may be expressed as the conditions $\smash{\tilde{\epsilon}_V \ll 1}$ and $\smash{\ab \tilde{\eta}_V \ab \ll 1}$. Both of these requirements are verified on the known exact scaling solutions with a single-field negative-definite exponential potential that is steep enough \cite{Khoury:2001zk}; in the $d$-dimensional case, one finds $\smash{w = -1 + 2 /[(d-1) \, \tilde{\epsilon}_V]}$. However, an analysis of tree-level type-II supergravities reveals a tension in finding small fast-roll parameters in string-theoretic constructions \cite{Uzawa:2018sal}, and we will comment more on this in ssec. \ref{ssec: ekpyrosis and string compactifications}. Our analysis does not rely on the fast-roll parameters, but rather we find constraints on the $w$-parameter by analyzing the universal properties of the time-dependent solutions to the field equations.

\subsection{Ekpyrosis and string compactifications} \label{ssec: ekpyrosis and string compactifications}
Our results can be clearly exploited in the context of the Swampland Program. Below, we list three important points and then we comment on their interplay.

\begin{enumerate}[label=\Alph*.]
    \item \label{string-theoretic expecations on ekpyrosis} In essence, the Swampland de Sitter Conjectures \cite{Obied:2018sgi, Ooguri:2018wrx} effectively constrain the slow-roll parameters of inflation in a way that abhors shallow positive potentials in the asymptotics of moduli space. However, partial evidence exists that the opposite applies to negative potentials, too: steep and negative potentials are also hardly found in explicit examples; see for instance refs. \cite{Uzawa:2018sal, Lehners:2018vgi}. As the latter could in principle realize ekyprosis, and given that ekpyrosis was claimed to lead to similar physics as inflation, it would seem conceivable to speculate that string-theoretic negative potentials may be too shallow to accommodate a slow contraction of the universe.\footnote{We are indebted to Thomas Van Riet for bringing these observations to our attention. These are also discussed in ref. \cite{meeus:mathesis}.} Of course, the classification of a potential as shallow or steep is a delicate one. A clearer way to phrase the above is that string theory gives order-1 parameters, which may turn out to be too large for inflation but also too small for ekpyrosis.
    \item \label{SFRC} In view of the Swampland Distance Conjecture \cite{Ooguri:2006in}, the covariant entropy bound \cite{Bousso:1999xy} on the towers of states becoming lighter and lighter on the way towards the moduli-space asymptotics supports the Swampland Fast-Roll Conjecture of ref. \cite{Bernardo:2021wnv}. This is the statement that the potential $V$ of an effective field theory that descends consistently from quantum gravity should satisfy the inequality
    \begin{equation} \label{fast-roll conjecture}
        - \dfrac{\der V}{\kappa_d V} \geq c
    \end{equation}
    in regions where the potential is negative, i.e. $V<0$, where $c > 0$ is an order-1 positive numerical constant and $\der V$ is the norm of the scalar potential gradient on the moduli space. At field values where eq. (\ref{fast-roll conjecture}) does not hold, such as at an anti-de Sitter minimum, ref. \cite{Bernardo:2021wnv} further proposes the condition $\smash{\der^2 V \geq c' \kappa_d^2 V}$ to be in place, where $c' > 0$ is another order-1 positive numerical constant and $\der^2 V$ is the maximum of the Hessian eigenvalues of the potential. Remarkably, the conjecture of eq. (\ref{fast-roll conjecture}) is qualitatively compatible with the ekpyrotic scenario, which requires steep negative potentials. It is further argued in ref. \cite{Bernardo:2021wnv} that this conjecture is compatible with the Transplanckian Censorship Conjecture \cite{Bedroya:2019snp}.
    \item \label{SATCC} Concerning negative potentials, by studying the asymptotic regime, too, ref. \cite{Andriot:2022brg} argues that a condition for a contracting universe to be in the string landscape is $\smash{a(t)/a(t_i) \geq \kappa_d^{\frac{d}{d-2}} \sqrt{-V(t_i)}}$, where $t_i$ is an initial time before $t$. This represents the Swampland Anti-Transplanckian Censorship Conjecture. For instance, in a parameterization of time with $t_i < t < 0$, scaling cosmologies with a single negative potential are such that $\smash{a(t)/a(t_i) = (t/t_i)^{4/[(d-2) \gamma^2]}}$. Therefore, if true, this condition effectively prevents the big crunch from happening, posing a lower bound on the final scale factor. Along with another condition on the second derivatives of the potential, the conjecture can be argued to pose the bound in eq. (\ref{fast-roll conjecture}) with $\smash{c = 2/(\sqrt{d-1} \sqrt{d-2})}$.
\end{enumerate}

Therefore, there seems to be a conflict between the general properties that one may expect of string compactifications and the Swampland Fast-Roll and Anti-Transplanckian Censorship Conjectures. An important point is however the following: our analysis regards the time-dependent solutions to the field equations, not the shape of the scalar potential per se (the shape of the potential is an off-shell property). In particular, while the string-theoretic expectation \ref{string-theoretic expecations on ekpyrosis} that we discussed may be seen as a property of time-dependent solutions, the two conjectures \ref{SFRC} and \ref{SATCC} are referred to the shape of the potential in the moduli-space asymptotics. As FLRW-cosmology with a negative potential implies a big crunch after evolution for a finite time, there may not be enough time for the fields to evolve towards asymptotic regions.\footnote{The potential can be written as a sum of exponentials if the fields are sufficiently close to the moduli-space boundary. However, our analysis assumes the effective field theory to be valid. Hence, it applies if initial conditions are in a region of moduli space where e.g. the towers of states that are expected to become light in view of the Distance Conjecture are still heavy. As the big crunch comes in a finite time, the field excursion may be small enough that such towers do not get light enough to affect the effective field theory.} Below we explore this further.

Because up until slightly before the big crunch the $w$-parameter is bounded both from below and from above in view of eqs. (\ref{optimal upper w-bound for negative-definite potentials}, \ref{optimal lower w-bound for negative-definite potentials}), we are in a strong position to constrain the feasibility of ekpyrosis in analytic terms, in a given theory. Our bounds provide us with an analytic handle to quantitatively assess whether a certain model has a steep-enough potential to allow for ekpyrosis. In particular, we can draw general conclusions on whether string compactifications provide such a kind of potentials and discuss the aspects that may imply that ekpyrosis is in the Swampland.

To start, we can make two very generic observations. On the one hand, it is apparent that more rolling fields tend to increase the $w$-parameter, by the lower bound in eq. (\ref{optimal lower w-bound for negative-definite potentials}). In fact, physically there are more contributions to the kinetic energy and mathematically the coupling vectors are higher dimensional, thus having a greater length. On the other hand, it is clear from eq. (\ref{optimal upper w-bound for negative-definite potentials}) that there is also an upper bound that, numerically, is not expected to be extremely large in string compactifications, as it is made up of sums of supposedly order-1 parameters, unless at least one of these parameters is not order-1 or there are parametrically-many fields. More in detail, all potentials generated perturbatively in string compactifications feature a canonical $d$-dimensional dilaton $\tdelta$ with an exponential coupling of a definite sign, i.e. \cite{Shiu:2023nph, Shiu:2023fhb}
\begin{equation}
    \gamma_{\tdelta} = - \dfrac{d}{\sqrt{d-2}} + \chi_{\mathrm{E}} \dfrac{\sqrt{d-2}}{2},
\end{equation}
where $\chi_{\mathrm{E}} \leq 2$ is the Euler number that weighs the string-coupling perturbative order via the string-worldsheet topology.\footnote{With negative potentials generated perturbatively, a rolling $d$-dimensional dilaton $\tdelta$ typically drives the theory towards strong coupling and/or small Einstein-frame volume. It turns out that $\tdelta$ is self-T-dual, while the Einstein-frame radion $\tomega$ is self-S-dual. Nevertheless, for an actual implementation of ekpyrosis, at some point before the big crunch a positive potential must kick in and prevent the fall of the scalars to negative infinity. This should also be implemented in a string-theoretic realization.} Therefore, at tree-level, the maximal dilaton coupling one has is $\smash{\ab \gamma_{\tdelta} \ab = 2 /\sqrt{d-2}}$,\footnote{Since the sources of RR-fields are D-branes and not fundamental strings \cite{Polchinski:1995mt}, in this classification RR-fluxes weigh like one-loop terms, with $\smash{\chi_{\mathrm{E}}=0}$. However, flux potentials are positive definite.} which by itself would not accommodate for the ekpyrosis condition of eq. (\ref{ekpyrosis condition}). Even one-loop Casimir energies, with $\smash{\ab \gamma_{\tdelta} \ab = d /\sqrt{d-2}}$, are still not steep enough to easily accommodate for the ekpyrosis condition in eq. (\ref{ekpyrosis condition}). In principle, a potential generated quantum-mechanically only at very high loop order and with all potential terms negative could realize ekpyrosis due to a large contribution coming from the dilaton. However, it is hard to imagine a model where this may be the case.

In view of our analysis, there is no conclusive evidence in favour or against the Swampland Fast-Roll and Anti-Transplanckian Censorship Conjectures. However, we have observed that it seems possible to circumvent an hypothesis that underlies the argument in favour of the conjecture proposed by refs. \cite{Bernardo:2021wnv, Andriot:2022brg}, namely the fact that the evolution is in the moduli-space asymptotics. It generally seems hard to realize an ekpyrotic universe in the perturbative regime of string theory: the ubiquitous dilaton coupling is not large enough, at least up to one-loop terms. Therefore, one would need to have several other fields rolling down the potential and then have a mechanism to prevent all of them from eventually driving the total potential to negative infinity, and/or fields with negative potentials and very large couplings generated perturbatively. In fact, we can embrace a constructive attitude and exploit our analytic bounds to constrain the presence of ekpyrosis in the string landscape.

\subsection{Examples}

A trivial example we can study is that of a string compactification on a space of positive curvature. In this case, the scalar potential reads \cite{Shiu:2023fhb}
\begin{align*}
    V = - K_{R} \, \e^{\frac{2}{\sqrt{d-2}} \, \kappa_d \tdelta - \frac{2}{\sqrt{10-d}} \, \kappa_d \tsigma},
\end{align*}
where $K_{R} > 0$ is a positive constant proportional to the curvature of the $(10-d)$-dimensional compactification space $\smash{\mathrm{K}_{10-d}}$, with $2 < d < 9$. Here, $\tdelta$ and $\tsigma$ are the canonically-normalized $d$-dimensional dilaton and string-frame volume, respectively. In this case, the $\smash{\Gamma_\star^a}$-vector is very short, i.e. $\smash{({\Gamma_{\star})^2} = [4/(d-2)] [8/(10-d)] < \Gamma_d^2}$ for $2<d<9$. Therefore, after enough time, the state-parameter will approach the kination value $w=1$, according to the bound in eq. (\ref{upper w-bound for negative-definite potentials}). In particular, the attractor solution has
\begin{align*}
    \dot{\tdelta}^2(t) + \dot{\tsigma}^2(t) & = \dfrac{1}{\kappa_d^2 t^2} \dfrac{d-2}{d-1}, \\
    H(t) & = \dfrac{1}{\displaystyle (d-1) \, t},
\end{align*}
with the cosmological time parameterized to be $t \leq 0$. As is clear from the shape of the potential, the fields evolve towards the values $\smash{\tdelta \sim \infty}$ and $\smash{\tsigma \sim -\infty}$, i.e. towards divergent string coupling (i.e. $g_s \sim + \infty$) and vanishing Einstein-frame volume ($\smash{\tilde{\mathrm{vol}} \, \mathrm{K}_{10-d} \sim 0^+}$). Therefore, the effective-field theory description breaks down at some point in the evolution towards the attractor solution.

Another trivial example we can formally study is that of an O$p$-plane wrapping a $(p+1-d)$-dimensional cycle inside the compactification manifold, with $p+1 \geq d$. In this case, the potential reads \cite{Shiu:2023fhb}
\begin{align*}
    V = - K_{\mathrm{O}p} \, \e^{\frac{d+2}{2 \sqrt{d-2}} \, \kappa_d \tdelta - \frac{d+8-2p}{2 \sqrt{10-d}} \kappa_d \tsigma}.
\end{align*}
which is also unsuitable for ekpyrosis, being $\smash{(\Gamma_\star)^2 = \Gamma_d^2 - (9-p)(p+1-d)/(10-d) < \Gamma_d^2}$. Therefore, the attractor has $w=1$.

A trivial example with terms of both signs that we can study is that of a string compactification on a space of positive curvature, together with RR-$q$-form flux, with $q \leq 10-d$. In this case, the scalar potential reads \cite{Shiu:2023fhb}
\begin{align*}
    V = & - K_{R} \, \e^{\frac{2}{\sqrt{d-2}} \, \kappa_d \tdelta - \frac{2}{\sqrt{10-d}} \, \kappa_d \tsigma} + \Lambda_{F_q} \, \e^{\frac{d}{\sqrt{d-2}} \, \kappa_d \tdelta - \frac{(10-d) - 2q}{\sqrt{10-d}} \, \kappa_d \tsigma},
\end{align*}
where $\smash{\Lambda_{F_q} > 0}$ is a positive constant proportional to the squared number of flux units. In the $\tdelta$-direction, because $\smash{\gamma_{-}^{\tdelta}=\Gamma_{-}^{\tdelta} = -2/\sqrt{d-2}}$ and $\smash{\gamma_{+}^{\tdelta}=\Gamma_{+}^{\tdelta} = -d/\sqrt{d-2}}$, we have $\smash{\gamma_\diamond^{\tdelta} = 0}$: this is because the exponential coefficients of the dilaton have the same sign, and so, because the potentials are of opposite signs, the dilaton sees a potential valley. In the $\smash{\tsigma}$-direction, as $\smash{\gamma_{-}^{\tsigma}=\Gamma_{-}^{\tsigma} = 2/\sqrt{10-d}}$ and $\smash{\gamma_{+}^{\tsigma}=\Gamma_{+}^{\tsigma} = [(10-d)-2q]/\sqrt{10-d}}$, we have $\smash{\gamma_\diamond^{\tsigma} = 2/\sqrt{10-d}}$ if $q \geq (8-d)/2$, or $\smash{\gamma_\diamond^{\tsigma} = 0}$ otherwise. In this case, the bound in eq. (\ref{lower w-bound}) is trivial, since $\smash{(\gamma_\diamond)^2 \leq \Gamma_d^2}$. No basis rotation can convincingly support ekpyrosis since there are only two fields and all couplings are around order-1 values. One would need knowledge of an upper bound, in order to possibly rule out ekpyrosis. However, the fact that none of the couplings is parametrically large provides a heuristic motivation to believe that the potential above cannot possibly give a very large parameter $w \gg 1$.

\section{Conclusions} \label{sec: conclusions}

In this article, we derive and discuss exact upper and lower bounds for the $w$-parameter of cosmological fluids consisting of canonical scalars with arbitrary negative multi-exponential potentials. To do so, we do not need to make any approximation on the field equations, we do not need to find an explicit solution to such equations, and we do not make reference to the initial conditions. Yet, the bounds apply to all possible fully-fledged time-dependent solutions of such equations. All these features are in analogy with the late-time bounds on cosmic acceleration of ref. \cite{Shiu:2023nph}.

The bounds we find are useful for assessing whether a given negative multi-field multi-exponential potential could hold sufficient properties for ekpyrosis. One can make use of these bounds both in phenomenological descriptions of ekpyrosis, where the couplings are arbitrary, and in the scan of the string landscape, testing the implications of the model-independent bounds on the couplings dictated by the higher-dimensional structure of a string compactification.

In the formulation of our bounds, we elucidated the geometric understanding of the limiting behaviors of the parameter $w$ that emerges from the scalar potential. This also provides one with a clear handle on the features that a theory must have in order to be compatible with ekpyrosis. Qualitatively, we observed an apparent mild tension of string compactifications with ekpyrosis. Working with canonical scalars, we can single out general patterns that would allow for ekpyrosis, such as requiring negative energies from high-order loop contributions, a large number of rolling scalar fields, and/or being away from the perturbative regime of string theory.

\acknowledgments
We are extremely grateful to Thomas Van Riet for collaborating in the initial stages of this project and for providing us with inspiring observations on ekpyrosis. We would also like to thank David Andriot, Caroline Jonas and George Tringas for several valuable discussions. GS is supported in part by the DOE grant DE-SC0017647. FT is supported by the FWO Odysseus grant GCD-D0989-G0F9516N. HVT is supported in part by the NSF CAREER grant DMS-1843320 and a Vilas Faculty Early-Career Investigator Award.

\appendix

\section{Bounds on ekpyrotic universes} \label{app: bounds on ekpyrotic universes}

Let $\smash{V = \sum_{i = 1}^m \Lambda_i \, \e^{- \kappa_d \gamma_{i a} \phi^a} = - \sum_{i = 1}^m K_i \, \e^{- \kappa_d \gamma_{i a} \phi^a}}$, with $K_i > 0$ for all terms $i=1,\dots,m$, be the scalar potential for the canonically-normalized scalar fields $\phi^a$, with $a=1,\dots,n$, in a $d$-dimensional FLRW-background: one can reduce the cosmological scalar-field and Friedmann equations to a system of autonomous equations \cite{Halliwell:1986ja, Copeland:1997et, Coley:1999mj, Guo:2003eu}. In terms of the variables
\begin{align*}
    x^a & = \dfrac{\kappa_d}{\sqrt{d-1} \sqrt{d-2}} \, \dfrac{\dot{\phi}^a}{H}, \\
    y^i & = \dfrac{\kappa_d \sqrt{2}}{\sqrt{d-1} \sqrt{d-2}} \, \dfrac{1}{H} \, \sqrt{\Lambda_i \, \e^{- \kappa_d \gamma_{i a} \phi^a}},
\end{align*}
and defining for simplicity
\begin{align*}
    f & = (d-1) H, \\
    c_{i a} & = \dfrac{1}{2} 
    \dfrac{\sqrt{d-2}}{\sqrt{d-1}} \, \gamma_{ia},
\end{align*}
the cosmological equations can indeed be expressed as
\begin{subequations}
\begin{align}
    \dot{x}^a & = \biggl[ -x^a (y)^2 + \sum_{i=1}^m {c_i}^{a} (y^i)^2 \biggr] \, f, \label{x-equation} \\
    \dot{y}^i & = \bigl[ (x)^2 - c_{ia} x^a \bigr] \, y^i f, \label{y-equation}
\end{align}
\end{subequations}
jointly with the two conditions
\begin{subequations}
\begin{align}
    & \dfrac{\dot{f}}{f^2} = - (x)^2, \label{f-equation} \\
    & (x)^2 + (y)^2 = 1. \label{sphere-condition}
\end{align}
\end{subequations}
(Here, the position of the $i$- and $a$-indices is arbitrary, since the former are just dummy labels and the latter refer to a field-space metric that is a Kronecker delta; if they appear together, such as in the coefficients $\smash{c_{ia}}$ and $\smash{{c_i}^a}$, we always write them on the left and on the right, respectively, to avoid confusion, and we place the $a$-index in the same (upper or lower) position as the other $x^a$-types variable that appear in the expression of reference. Einstein summation convention is understood for the metric-contraction on the $a$-indices and moreover we use the shorthand notations $\smash{(x)^2 = x_a x^a}$ and $\smash{(y)^2 = \sum_{i=1}^m (y^i)^2}$.)

By defining $(z^i)^2 = - (y^i)^2$ and $g = -f$, one obtains the dynamical system
\begin{subequations}
\begin{align}
    \dot{x}^a & = \biggl[ -x^a (z)^2 + \sum_{i=1}^m {c_i}^{a} (z^i)^2 \biggr] \, g, \label{x-equation (E)} \\
    \dot{z}^i & = - \bigl[ (x)^2 - c_{ia} x^a \bigr] \, z^i g, \label{z-equation (E)}
\end{align}
\end{subequations}
jointly with the two conditions
\begin{subequations}
\begin{align}
    & \dfrac{\dot{g}}{g^2} = + (x)^2, \label{g-equation (E)} \\
    & (x)^2 - (z)^2 = 1. \label{hyperbola-condition (E)}
\end{align}
\end{subequations}

Below are a series of mathematical results. A formulation in terms of physical observables and an analysis of their implications is in the main text. The discussion shares similarities with the derivation of late-time bounds on $(x)^2$ in the presence of a positive $(y^2)$-term performed in ref. \cite{Shiu:2023nph}, but it corresponds to a completely different physical and mathematical scenario. \\

Let the unknown functions be such that $x^a \in  \, \mathbb{R}$ and $z^i \in \, \mathbb{R}$; let $c^a$ and $C^a$ denote the minimum and the maximum of the constant parameters $\smash{{c_i}^a}$ for each $a$-index, i.e. $\smash{c^a = \min_i {c_i}^a}$ and $\smash{C^a = \max_i {c_i}^a}$; let $C_\star^a = \max \lbrace |c^a|, |C^a| \rbrace > 0$. Let $t_0$ be the initial time and let $g(t_0) > 0$. Let $\smash{\varphi(t) = \int_{t_0}^t \de s \; g(s) \bigl[ z(s) \bigr]^2}$.

\vspace{8pt}
\begin{lemma} \label{lemma: g - growing function (E)}
    At all times $t > t_0$ one has
    \begin{equation} \label{eq.: g - growing function (E)}
        g(t) \geq g(t_0) > 0.
    \end{equation}
\end{lemma}

\begin{proof}
In view of eq. (\ref{g-equation (E)}), one has $\smash{\dot{g}(t) \geq 0}$, which means that $g(t)$ is never decreasing.
\end{proof}

\vspace{8pt}
\begin{lemma} \label{lemma: g - lower bound (E)}
    At all times $t > t_0$ one has
    \begin{equation} \label{eq.: g - lower bound (E)}
        g(t) \geq \dfrac{1}{- (t - t_0) + \dfrac{1}{g(t_0)}}.
    \end{equation}
\end{lemma}

\begin{proof}
In view of eqs. (\ref{g-equation (E)}, \ref{hyperbola-condition (E)}), since $\bigl[ x(t) \bigr]^2 \geq 1$ at all times $t$, one has $\smash{- \de \bigl[ 1/g(t) \bigr] / \de t \geq 1}$, which can be integrated to give $\smash{1/g(t) \geq -(t - t_0) + 1/g(t_0)}$. Because $g(t) > 0$ by eq. (\ref{eq.: g - growing function (E)}), one immediately gets eq. (\ref{eq.: g - lower bound (E)}).
\end{proof}

\vspace{8pt}
\begin{corollary} \label{corollary: g-divergence (E)}
    As $g(t_0) > 0$, there exists a finite time $t_\star$, belonging in the interval $\smash{t_\star \in \, ]t_0, t_0 + 1/g(t_0)]}$, at which the function $g(t)$ diverges, i.e.
    \begin{equation} \label{eq.: g-divergence (E)}
        \lim_{t \to t_\star} g(t) = + \infty.
    \end{equation}
\end{corollary}

\begin{proof}
    In view of eq. (\ref{eq.: g - lower bound (E)}), the function $g(t)$ is bounded from below by a function that diverges at $t = t_0 + 1/g(t_0)$. Hence, eq. (\ref{eq.: g-divergence (E)}) follows immediately.
\end{proof}

\vspace{12pt}

In view of corollary \ref{corollary: g-divergence (E)}, it is convenient to perform a change of variables to avoid dealing with a divergent function. Let $u: \; \mathbb{R}_0^+ \, \to \, [t_0, t_\star]$ be the function such that
\begin{equation}
    u'(\tau) = \dfrac{\de u}{\de \tau} (\tau) = \dfrac{1}{g [u(\tau)]},
\end{equation}
together with the initial condition $u(0) = t_0$. A function $\mathsf{h}(\tau) = h [u(\tau)]$ is such that the derivative with respect to $\tau$ reads $\smash{\mathsf{h}'(\tau) = u'(\tau) \cdot \de h[u(\tau)]/\de u = \bigl( 1 / g[u(\tau)] \bigr) \cdot \de h[u(\tau)]/\de u}$. For the variables (notice the different notation between $\smash{(x^a,z^i)}$ and $\smash{(\X^a,\Z^i)}$ from now on)
\begin{subequations}
\begin{align}
    \X^a(\tau) & = x^a [u(\tau)], \\
    \Z^i(\tau) & = z^i [u(\tau)],
\end{align}
\end{subequations}
the cosmological autonomous system in eqs. (\ref{x-equation (E)}, \ref{z-equation (E)}) takes the form
\begin{subequations}
\begin{align}
    \X'{}^a & = -\X^a (\Z)^2 + \sum_{i=1}^m {c_i}^{a} (\Z^i)^2, \label{X-equation (E)} \\
    \Z'{}^i & =- \bigl[(\X)^2 - c_{ia} \X^a \bigr] \, \Z^i, \label{Z-equation (E)}
\end{align}
\end{subequations}
with the constraint in eq. (\ref{hyperbola-condition (E)}) being
\begin{equation} \label{(X,Y)-plane hyperobola-condition (E)}
    (\X)^2 - (\Z)^2 = 1.
\end{equation}
Let $\smash{G(t) = \int_{t_0}^t \de s \, g(s)}$. Then, one has $\smash{\de G[u(\tau)] / \de \tau = 1}$, which means that $\smash{G[u(\tau)] = \tau}$. Because then $\smash{u(\tau) = G^{-1} (\tau)}$, the function $u=u(\tau)$ is indeed a function $u: \; [t_0, \infty[ \, \to \, [t_0, t_\star]$. In particular, it is such that
\begin{align*}
    \lim_{\tau \to \infty} u(\tau) = t_\star.
\end{align*}
As integrations give $\smash{\int_{\tau_1}^{\tau_2} \de \tau \, \mathsf{h}(\tau) = \int_{u(\tau_1)}^{u(\tau_2)} \de s \, g(s) \, h(s)}$, one has $\varphi(t) =  \int_{t_0}^t \de s \; g(s) \bigl[ z(s) \bigr]^2 = \int_{0}^{u^{-1}(t)} \de \tau \, [\Z(\tau)]^2$. One can also write $\smash{\textnormal{\textphi}(\tau) = \varphi[u(\tau)] = \int_{0}^{\tau} \de \sigma \, \bigl[ \Z(\sigma) \bigr]^2}$.

\vspace{12pt}

\begin{lemma} \label{lemma: X - upper bound (E)}
    If $\smash{\lim_{\tau \to \infty} \textnormal{\textphi}(\tau) = \infty}$, then one has
    \begin{equation} \label{eq.: X - upper bound (E)}
        \limsup_{\tau \to \infty} \X^a(\tau) \leq C^a.
    \end{equation}
\end{lemma}

\begin{proof}
    In view of eq. (\ref{X-equation (E)}), one can write the inequality
    \begin{align*}
        \X'{}^a = - \X^a (z)^2 + \sum_{i} {c_{i}}^{a} (\Z^i)^2 \leq \bigl[ -\X^a + C^a \bigr] \, (\Z)^2.
    \end{align*}
    Then, one can write
    \begin{align*}
        \dfrac{\de}{\de \tau} \, \bigl[ \e^{\textnormal{\textphi}(\tau)} \X^a(t) \bigr] = \e^{\textnormal{\textphi}(\tau)} \bigl[ \X'{}^a(\tau) + \X^a(\tau) \, [\Z(\tau)]^2 \bigr] \leq C^a \, \e^{\textnormal{\textphi}(\tau)} [\Z(\tau)]^2.
    \end{align*}
    By integrating this inequality, one then finds the further inequality
    \begin{align*}
        \e^{\textnormal{\textphi}(\tau)} \X^a(\tau) - \X^a(0) \leq C^a \int_{0}^\tau \de \sigma \; \e^{\textnormal{\textphi}(\sigma)} [\Z(\sigma)]^2 = C^a \, \bigl[ \e^{\textnormal{\textphi}(\tau)} - 1 \bigr].
    \end{align*}
    If $\smash{\lim_{\tau \to \infty} \textnormal{\textphi}(\tau) = \infty}$, then the conclusion in eq. (\ref{eq.: X - upper bound (E)}) holds.
\end{proof}

\vspace{8pt}
\begin{lemma} \label{lemma: X - lower bound (E)}
    If $\smash{\lim_{\tau \to \infty} \textnormal{\textphi}(\tau) = \infty}$, then one has
    \begin{equation} \label{eq.: X - lower bound (E)}
        \liminf_{\tau \to \infty} \X^a(\tau) \geq c^a.
    \end{equation}
\end{lemma}

\begin{proof}
    In view of eq. (\ref{X-equation (E)}), one can write the inequality
    \begin{align*}
        \X'{}^a = -\X^a (\Z)^2 + \sum_{i} {c_{i}}^{a} (\Z^i)^2 \geq \bigl[ -\X^a + c^a \bigr] (\Z)^2.
    \end{align*}
    Then, one can write
    \begin{align*}
        \dfrac{\de}{\de \tau} \, \bigl[ \e^{\textnormal{\textphi}(\tau)} \X^a(\tau) \bigr] = \e^{\textnormal{\textphi}(\tau)} \bigl[ \X'{}^a(\tau) + \X^a(\tau) \, \bigl[ \Z(\tau) \bigr]^2 \bigr] \geq c^a \, \e^{\textnormal{\textphi}(\tau)} \bigl[ \Z(t) \bigr]^2.
    \end{align*}
    By integrating this inequality, one then finds the further inequality
    \begin{align*}
        \e^{\textnormal{\textphi}(\tau)} \X^a(\tau) - \X^a(0) \geq c^a \int_{0}^\tau \de \sigma \; \e^{\textnormal{\textphi}(\sigma)} \bigl[ \Z(\sigma) \bigr]^2 = c^a \, \bigl[ \e^{\textnormal{\textphi}(\tau)} - 1 \bigr].
    \end{align*}
    If $\smash{\lim_{\tau \to \infty} \textnormal{\textphi}(\tau) = \infty}$, then the conclusion in eq. (\ref{eq.: X - lower bound (E)}) holds.
\end{proof}

\vspace{8pt}
\begin{remark}
    In view of eqs. (\ref{X-equation (E)}, \ref{Z-equation (E)}), by defining $\smash{\X'{}^a = - \X^a}$, one can reformulate the problem in terms of the set of constants $\smash{c'_i{}^a = - c_i{}^a}$. Because $\smash{\sup \X^a = - \inf \X'{}^a}$ and $\smash{c'{}^a = \min_i c'_i{}^a = - \max_i {c_i}^a = - C^a}$, the constraint $\smash{\liminf_{\tau \to \infty} \X^a(\tau) \geq c^a}$ in eq. (\ref{eq.: X - lower bound (E)}) is equivalent to the constraint $\smash{\limsup_{\tau \to \infty} x'{}^a(\tau) \leq C'{}^a}$ in eq. (\ref{eq.: X - upper bound (E)}).
\end{remark}

\vspace{8pt}
\begin{remark}
    In the proofs of lemmas \ref{lemma: X - upper bound (E)} and \ref{lemma: X - lower bound (E)}, two important inequalities have been derived. Independently of the value of $\smash{\lim_{\tau \to \infty} \textnormal{\textphi}(\tau)}$, one has
    \begin{subequations}
        \begin{align}
            \X^a(\tau) & \leq C^a + \e^{-\textnormal{\textphi}(\tau)} \bigl[ \X^a(0) - C^a \bigr], \label{eq.: x universal upper bound} \\
            \X^a(\tau) & \geq c^a + \e^{-\textnormal{\textphi}(\tau)} \bigl[ \X^a(0) - c^a \bigr]. \label{eq.: x universal lower bound}
        \end{align}
    \end{subequations}
\end{remark}

\vspace{8pt}
\begin{theorem} \label{theorem: X^2-bound (E)}
    If $\smash{\lim_{\tau \to \infty} \textnormal{\textphi}(\tau) = \infty}$, then one has
    \begin{equation} \label{eq.: X^2-bound (E)}
        \limsup_{\tau \to \infty} \bigl[ \X(\tau) \bigr]^2 \leq (C_\star)^2.
    \end{equation}
\end{theorem}

\begin{proof}
    By eqs. (\ref{eq.: X - upper bound (E)}, \ref{eq.: X - lower bound (E)}), it is obvious that one has $\smash{\limsup_{\tau \to \infty} \ab \X^a(\tau) \ab \leq C_\star^a}$. Because $\smash{\ab \X^a(\tau) \ab \geq 0}$, eq. (\ref{eq.: X^2-bound (E)}) follows immediately.
\end{proof}

\vspace{8pt}
\begin{corollary}
    If $\smash{(C_\star)^2 < 1}$, then $\smash{\lim_{\tau \to \infty} \textnormal{\textphi}(\tau) < \infty}$.
\end{corollary}

\begin{proof}
    By contradiction, let $\smash{\lim_{\tau \to \infty} \textnormal{\textphi}(\tau) = \infty}$. Then one can write $\smash{(\X)^2 \leq (C_\star)^2 < 1}$, which is impossible by eq. (\ref{(X,Y)-plane hyperobola-condition (E)}). Therefore, one cannot have $\smash{\lim_{\tau \to \infty} \textnormal{\textphi}(\tau) = \infty}$.
\end{proof}

\vspace{8pt}
\begin{theorem} \label{theorem: non-proper attractors (E)}
    If $\smash{\lim_{\tau \to \infty} \textnormal{\textphi}(\tau) < \infty}$, then
    \begin{subequations}
    \begin{align}
        \lim_{\tau \to \infty} \X^a(\tau) & = \tilde{\X}^a,  \label{eq.: unique non-proper X-attractor (E)} \\
        \lim_{\tau \to \infty} \Z^i(\tau) & = 0,  \label{eq.: unique non-proper Z-attractor (E)}
    \end{align}
    \end{subequations}
    where, in particular, one has $\smash{(\tilde{\X})^2 = 1}$.
\end{theorem}

\begin{proof}
    By the bounds in eqs. (\ref{eq.: x universal upper bound}, \ref{eq.: x universal lower bound}), $\X^a$ is such that $\ab \X^a \ab$ is universally bounded as $|\X^a| <C_\star^a+|\X^a(0)|$.
    
    For any given $a$-index, the inequalities hold
    \begin{align*}
        \bigl[ -\X^a + C^a \bigr] (\Z)^2 \geq \X'{}^a \geq \bigl[ -\X^a + c^a \bigr] (\Z)^2.
    \end{align*}
    If $\smash{\X'{}^a \geq 0}$, then
    \begin{align*}
        \bigl\ab \X'{}^a \bigr\ab \leq \bigl\ab -\X^a + C^a \bigr\ab (\Z)^2 \leq \bigl[ \ab \X^a \ab + \ab C^a \ab \bigr] (\Z)^2.
    \end{align*}
    If $\smash{\X'{}^a < 0}$, then
    \begin{align*}
        \bigl\ab \X'{}^a \bigr\ab \leq \bigl\ab -\X^a + c^a \bigr\ab (\Z)^2 \leq \bigl[ \ab \X^a \ab + \ab c^a \ab \bigr] (\Z)^2.
    \end{align*}
    Therefore, one can write
    \begin{align*}
        \bigl\ab \X'{}^a \bigr\ab \leq \bigl[ \ab \X^a \ab + C_\star^a \bigr] \, (\Z)^2 \leq \bigl[ \ab \X^a(0) \ab + 2 C_\star^a \bigr] \, (\Z)^2.
    \end{align*}
    Because $\smash{\lim_{\tau \to \infty} \textnormal{\textphi}(\tau) = \lim_{\tau \to \infty} \int_{0}^\tau \de \sigma \; \bigl[ \Z(\sigma) \bigr]^2 < \infty}$, for any arbitrary constant $\epsilon^a > 0$, there exists a time $\smash{\tau_{\epsilon^a} > 0}$ such that
    \begin{align*}
        \int_{\tau_{\epsilon^a}}^\infty \de \sigma \; \bigl[ \Z(\sigma) \bigr]^2 < \dfrac{\epsilon^a}{\ab \X^a(0) \ab + 2 C_\star^a}.
    \end{align*}
    Therefore, for a time $\tau_*$ in the range $\smash{\tau > \tau_* >  \tau_{\epsilon^a}}$, one can write
    \begin{align*}
        \bigl\ab \X^a(\tau) - \X^a(\tau_*) \bigr\ab = \biggl\ab \int_{\tau_*}^\tau \de \sigma \; \X'{}^a(\sigma) \biggr\ab \leq \int_{\tau_*}^\tau \de \sigma \; \bigl\ab \X'{}^a(\sigma) \bigr\ab \leq \int_{\tau_*}^\tau \de \sigma \; \bigl[ \ab \X^a(0) \ab + 2 C_\star^a \bigr] \, \bigl[ \Z(\sigma) \bigr]^2 < \epsilon^a.
    \end{align*}
    This means that $\smash{\lbrace \X^a(\tau) \rbrace_{\tau >0}}$ is a Cauchy sequence, which implies the existence of the limit
    \begin{align*}
        \lim_{\tau \to \infty} \X^a(\tau) = \tilde{\X}^a.
    \end{align*}
    Furthermore, this implies the existence of the limit $\smash{\lim_{\tau \to \infty} \bigl[\X(\tau)\bigr]^2 = (\tilde{\X})^2}$. Therefore, for any arbitrary constant $\delta > 0$, there exists a time $\smash{\tau_{\delta} > 0}$ such that $\smash{\bigl\ab \bigl[\X(\tau)\bigr]^2 - (\tilde{\X})^2 \bigr\ab < \delta}$ at all times $\smash{\tau > \tau_\delta}$. For all times $\tau > \tau_\delta$, one can thus write the inequality
    \begin{align*}
        \bigl[ \X(\tau) \bigr]^2 > (\tilde{\X})^2 - \delta.
    \end{align*}
    Assuming the condition $\smash{(\tilde{\X})^2 > 1}$ to hold by contradiction, then, by choosing $\delta$ such that $\delta < (\tilde{\X})^2 - 1$, one can write
    \begin{align*}
        \int_{\tau_\delta}^\infty \de \sigma \; \bigl[ \Z(\sigma) \bigr]^2 = \int_{\tau_\delta}^\infty \de \sigma \; \bigl[ \bigl[\X(\sigma)\bigr]^2 - 1 \bigr] > \bigl[ (\tilde{\X})^2 - \delta - 1 \bigr]\int_{\tau_\delta}^\infty \de \sigma = \infty,
    \end{align*}
    which is impossible. Hence, one has $\smash{(\tilde{\X})^2=1}$.
\end{proof}

\vspace{8pt}
\begin{remark}
    Although $\smash{(\X)^2 \geq 1}$, the individual terms $\smash{\X^a}$ are not constrained in any way. Moreover, the fact that $\smash{(\X)^2}$ is bounded implies that $\smash{(\Z)^2}$ is also bounded.
\end{remark}

\vspace{8pt}
\begin{corollary} \label{corollary: (x)^2-bounds (E)}
    If $\smash{(C_\star)^2 \geq 1}$ and if $\smash{\lim_{t \to t_\star} \varphi(t) = \infty}$, then one has
    \begin{equation}
        1 \leq \limsup_{t \to t_\star} \bigl[ x(t) \bigr]^2 \leq (C_\star)^2;
    \end{equation}
    else, one has
    \begin{equation}
        \lim_{t \to t_\star} \bigl[ x(t) \bigr]^2 = 1.
    \end{equation}
\end{corollary}

\begin{proof}
    Because $\smash{\lim_{\tau \to \infty} u(\tau) = t_\star}$, one has $\smash{\lim_{\tau \to \infty} \mathsf{h}(\tau) = \lim_{t \to t_\star} h(t)}$ for any function $\smash{\mathsf{h}(\tau) = h [u(\tau)]}$. Hence, the results of theorems \ref{theorem: X^2-bound (E)} and \ref{theorem: non-proper attractors (E)} translate immediately from the variable $\tau$ to the variable $t$.
\end{proof}

\vspace{12pt}
As the functions $\X^a(\tau)$ are limited from above and below according to lemmas \ref{lemma: X - upper bound (E)} and \ref{lemma: X - lower bound (E)}, along with the conclusions of theorem \ref{theorem: non-proper attractors (E)} too, formally the same identical conclusions as in ref. \cite{Shiu:2023nph} for the lower bound for $[\X(\tau)]^2$, and hence for $[x(t)]^2$, also hold. For each $a$-index, if $c^a > 0$, let $\smash{c_\diamond^a = c^a}$; if $c^a \leq 0$, then let $\smash{c_\diamond^a = 0}$.

\vspace{8pt}
\begin{corollary} \label{corollary: (x)^2 lower bound (E)}
    If $\smash{\lim_{t \to t_\star} \varphi(t) = \infty}$, then one has
    \begin{equation}
        \liminf_{t \to t_\star} \bigl[ x(t) \bigr]^2 \geq (c_\diamond)^2,
    \end{equation}
    which is trivial if $\smash{(c_\diamond)^2 < 1}$; else, one has
    \begin{equation}
        \lim_{t \to t_\star} \bigl[ x(t) \bigr]^2 = 1.
    \end{equation}
\end{corollary}

\begin{proof}
    The proof follows straightforwardly from lemmas \ref{lemma: X - upper bound (E)} and \ref{lemma: X - lower bound (E)} and theorem \ref{theorem: non-proper attractors (E)}, after the change of variable from $\tau$ to $t$.
\end{proof}

\vspace{8pt}
\begin{remark}
    It should be noticed that, by construction, $\smash{(c_\diamond)^2 \leq (C_\star)^2}$. Therefore, corollaries \ref{corollary: (x)^2-bounds (E)} and \ref{corollary: (x)^2 lower bound (E)} are never in contradiction.
\end{remark}

\vspace{12pt}

Part of the above results can be generalized to the case in which some of the $z^i$-functions are imaginary. Let $\smash{(z^{i_-})^2 = (z_-^{i_-})^2 > 0}$ and $\smash{(z^{i_+})^2 = - (z_+^{i_+})^2 > 0}$, distinguishing over all possible indices $\smash{i = i_-, i_+}$. For each $a$-index, let $\smash{c_-^a = \min_{i_-} {c_{i_-}}^a}$ and $\smash{C_+^a = \max_{i_+} {c_{i_+}}^a}$, and let $\smash{C_-^a = \max_{i_-} {c_{i_-}}^a}$ and $\smash{c_+^a = \min_{i_+} {c_{i_+}}^a}$. As a further assumption, let $\smash{(z_-)^2 > (z_+)^2}$ at all times. It is apparent that lemmas \ref{lemma: g - growing function (E)} and \ref{lemma: g - lower bound (E)} as well as corollary \ref{corollary: g-divergence (E)} hold true. Therefore, one can also reformulate the dynamical system in terms of the $\tau$-variable, generalizing trivially eqs. (\ref{X-equation (E)}, \ref{Z-equation (E)}, \ref{(X,Y)-plane hyperobola-condition (E)}).

\vspace{8pt}
\begin{lemma} \label{lemma: X - upper bound 2 (E)}
    If $\smash{c_+^a \geq C_-^a}$ and if $\smash{\lim_{\tau \to \infty} \textnormal{\textphi}(\tau) = \infty}$, then one has
    \begin{equation} \label{eq.: X - upper bound 2 (E)}
        \limsup_{\tau \to \infty} \X^a(\tau) \leq c_+^a.
    \end{equation}
\end{lemma}

\begin{proof}
    In view of eq. (\ref{X-equation (E)}) and of the inequality $\smash{c_+^a \geq C_-^a}$, one can write the inequality
    \begin{align*}
        \X'{}^a \leq - \X^a (z)^2 + C_-^{a} (\Z_-)^2 - c_+^{a} (\Z_+)^2 \leq \bigl[ -\X^a + c_+^a \bigr] \, (\Z)^2
    \end{align*}
    and replicate all the steps that lead to the proof of lemma \ref{lemma: X - upper bound (E)}.
\end{proof}

\vspace{8pt}
\begin{lemma} \label{lemma: X - lower bound 2 (E)}
    If $\smash{c_-^a \geq C_+^a}$ and if $\smash{\lim_{\tau \to \infty} \textnormal{\textphi}(\tau) = \infty}$, then one has
    \begin{equation} \label{eq.: X - lower bound 2 (E)}
        \liminf_{\tau \to \infty} \X^a(\tau) \geq c_-^a.
    \end{equation}
\end{lemma}

\begin{proof}
    In view of eq. (\ref{X-equation (E)}) and of the inequality $\smash{c_-^a \geq C_+^a}$, one can write the inequality
    \begin{align*}
        \X'{}^a \geq - \X^a (z)^2 + c_-^{a} (\Z_-)^2 - C_+^{a} (\Z_+)^2 \geq \bigl[ -\X^a + c_-^a \bigr] (\Z)^2
    \end{align*}
    and replicate all the steps that lead to the proof of lemma \ref{lemma: X - lower bound (E)}.
\end{proof}

\vspace{12pt}

As the functions $\X^a(\tau)$ are limited from above and below according to lemmas \ref{lemma: X - upper bound 2 (E)} and \ref{lemma: X - lower bound 2 (E)}, formally the same identical conclusions as in ref. \cite{Shiu:2023nph} for the lower bound for $[\X(\tau)]^2$, and hence for $[x(t)]^2$, also hold. For each $a$-index, let $\smash{c_-^a \geq C_+^a}$. If $\smash{c_-^a > 0}$, let $\smash{c_\diamond^a = c_-^a}$; if $\smash{c_-^a \leq 0}$, then let $\smash{c_\diamond^a = 0}$. If the condition $\smash{c_-^a \geq C_+^a}$ is not verified even after the redefinition $\smash{x'{}^a = - x^a}$, let instead $\smash{c_\diamond^a = 0}$.

\vspace{8pt}
\begin{corollary} \label{corollary: (x)^2 lower bound 2 (E)}
    If $\smash{\lim_{t \to t_\star} \varphi(t) = \infty}$, then one has
    \begin{equation}
        \liminf_{t \to t_\star} \bigl[ x(t) \bigr]^2 \geq (c_\diamond)^2,
    \end{equation}
    which is trivial if $\smash{(c_\diamond)^2 < 1}$.
\end{corollary}

\begin{proof}
    The proof follows straightforwardly from lemmas \ref{lemma: X - upper bound 2 (E)} and \ref{lemma: X - lower bound 2 (E)} after the change of variable from $\tau$ to $t$.
\end{proof}

\section{Negative-potential scaling solutions} \label{app: negative-potential scaling solutions}
A single-field negative-definite exponential potential $\smash{V = - K \, \e^{- \kappa_d \gamma \phi}}$ admits a scaling solution which is simply a critical point of the associated dynamical system; for more details, see app. \ref{app: bounds on ekpyrotic universes}. If $\smash{\gamma > 2 \, \sqrt{d-1}/\sqrt{d-2}}$, the solution is
\begin{subequations}
    \begin{align}
        \phi(t) & = \tilde{\phi} + \dfrac{1}{\kappa_d} \dfrac{2}{\gamma} \; \mathrm{ln} \, \biggl[ 1 - \dfrac{d-2}{4} \, \gamma^2 \tilde{H} (\tilde{t} - t) \biggr], \\
        H(t) & = \dfrac{1}{\displaystyle - \dfrac{d-2}{4} \, \gamma^2 (\tilde{t} - t) + \dfrac{1}{\tilde{H}}},
    \end{align}
\end{subequations}
where $t < \tilde{t}$. This corresponds to the collapsing scale factor
\begin{align*}
    a(t) = \tilde{a} \, \biggl[ 1 - \dfrac{d-2}{4} \gamma^2 \tilde{H} (\tilde{t} - t) \biggr]^{\frac{4}{d-2} \frac{1}{\gamma^2}},
\end{align*}
which crunches at the time $\smash{t_\star = \tilde{t} - 4 / [(d-2) \gamma^2 \tilde{H}]}$; approaching this time, one gets the singular behaviors $\smash{\phi(t) \!\overset{t \sim t_\star^-\!\!\!}{\sim}\! -\infty}$ and $\smash{H(t) \!\overset{t \sim t_\star^-\!\!\!}{\sim}\! -\infty}$. One also finds the compatibility condition
\begin{align*}
    \dfrac{2 \kappa_d^2 \, K}{(d-1)(d-2)} \,\e^{-\kappa_d \gamma \tilde{\phi}} = \biggl[ \dfrac{1}{4} \dfrac{d-2}{d-1} \, \gamma^2 - 1 \biggr] \tilde{H}^2.
\end{align*}
If $\smash{\gamma \leq 2 \, \sqrt{d-1}/\sqrt{d-2}}$, the solution is
\begin{subequations}
    \begin{align}
        \phi(t) & = \tilde{\phi} + \dfrac{1}{\kappa_d} \, \dfrac{\sqrt{d-2}}{\sqrt{d-1}} \; \mathrm{ln} \, \bigl[ 1 - (d-1) \tilde{H} (\tilde{t} - t) \bigr], \\
        H(t) & = \dfrac{1}{\displaystyle - (d-1) (\tilde{t} - t) + \dfrac{1}{\tilde{H}}}.
    \end{align}
\end{subequations}
which corresponds to a formally vanishing potential. In this case, the solution is a universal attractor.

\section{Energy density of curvature anisotropies} \label{app: curvature anisotropies}
Let a $d$-dimensional metric of the Kasner form be parameterized as
\begin{equation} \label{Kasner-type metric}
    ds^2_{1,d-1} = \tilde{g}_{\mu \nu} \de x^\mu \de x^\nu = - \de t^2 + a^2(t) \sum_{q=1}^n \e^{2 \beta_q(t)} \de l_{\mathbb{E}_{n_q}}^2,
\end{equation}
where $\smash{\sum_{q=1}^n n_q = d-1}$, subject to the Kasner-type condition
\begin{equation} \label{Kasner-type condition}
    \sum_{q=1}^n \beta_q = 0.
\end{equation}
This metric represents a deformation of the flat FLRW-metric in which the functions $\beta_q$ parameterize background anisotropies.

In order to compute how the anisotropies appear in the Einstein equations, it is convenient to work in the parameterization
\begin{align*}
    ds^2_{1,d-1} = - \de t^2 + \sum_{q=1}^n \e^{2 \gamma_q(t)} \de l_{\mathbb{E}_{n_q}}^2,
\end{align*}
where $\gamma_q = \alpha + \beta_q$, with $\alpha = \mathrm{ln} \, a$. One easily finds that the $\smash{\tilde{g}_{\mu \nu}}$-metric Ricci tensor has the non-zero components
\begin{align*}
    \tilde{R}_{00} & = - \sum_{q=1}^n n_q \bigl[ \ddot{\gamma}_q + (\dot{\gamma}_q)^2 \bigr], \\
    \tilde{R}_{i_q j_q} & = \sum_{q=1}^n n_q \biggl[ \ddot{\gamma}_q + \dot{\gamma}_q \sum_{p=1}^n n_p \dot{\gamma}_p \biggr] \, \e^{2 \gamma_q} \delta_{i_q j_q},
\end{align*}
with the associated Ricci scalar
\begin{align*}
    \tilde{R} = \sum_{q=1}^n n_q \biggl[ 2 \ddot{\gamma}_q + (\dot{\gamma}_q)^2 + \dot{\gamma}_q \sum_{p=1}^n n_p \dot{\gamma}_p \biggr].
\end{align*}
After imposing the Kasner-type conditions, these expressions greatly simplify and one finds
\begin{align*}
    \tilde{R}_{00} & = - (d-1) \bigl[ \ddot{\alpha} + (\dot{\alpha})^2 \bigr] - \sum_{q=1}^n n_q (\dot{\beta}_q)^2, \\
    \tilde{R}_{i_q j_q} & = \biggl[ \ddot{\alpha} + (d-1) (\dot{\alpha})^2 + \ddot{\beta}_q + (d-1) \dot{\alpha} \dot{\beta}_q \biggr] \, \e^{2 \alpha + 2 \beta_q} \delta_{i_q j_q},
\end{align*}
with the associated Ricci scalar
\begin{align*}
    \tilde{R} = 2 (d-1) \ddot{\alpha} + d (d-1) (\dot{\alpha})^2 + \sum_{q=1}^n n_q (\dot{\beta}_q)^2.
\end{align*}
Therefore, the Einstein tensor $\smash{\tilde{E}_{\mu \nu} = \tilde{R}_{\mu \nu} - \tilde{R} \, \tilde{g}_{\mu \nu}/2}$ takes the form
\begin{subequations}
    \begin{align}
        \tilde{E}_{00} & = \dfrac{(d-1)(d-2)}{2} H^2 - \dfrac{1}{2} \sum_{q=1}^n n_q (\dot{\beta}_q)^2, \\
        \tilde{E}_{i_q j_q} & = \biggl[ - (d-2) \dot{H} - \dfrac{(d-1)(d-2)}{2} H^2 + \ddot{\beta}_q + (d-1) H \dot{\beta}_q \biggr] \, a^2 \, \e^{2 \beta_q} \delta_{i_q j_q}.
    \end{align}
\end{subequations}
For consistency, each anisotropy term must satisfy the equation
\begin{align*}
    \ddot{\beta}_q + (d-1) H \dot{\beta}_q = 0,
\end{align*}
which means it evolves with respect to the scale factor as
\begin{align*}
    \beta_q (t) = \beta_q (t_0) \, \biggl( \dfrac{a(t_0)}{a(t)} \biggr)^{d-1}.
\end{align*}
Therefore, the energy density associated to the anisotropies can be computed exactly from the equation
\begin{align*}
    \dfrac{(d-1)(d-2)}{2} H^2 - \dfrac{1}{2} \sum_{q=1}^n n_q (\dot{\beta}_q)^2 = \dfrac{(d-1)(d-2)}{2} H^2 - \biggl( \dfrac{a(t_0)}{a(t)} \biggr)^{2(d-1)} \, \dfrac{1}{2} \sum_{q=1}^n n_q (\dot{\beta}_q(t_0))^2
\end{align*}
In conclusion, the $w$-parameter associated to all FLRW-metric anisotropies is $w=1$, in any dimension $d$. This agrees with the 4-dimensional result of ref. \cite{Lehners:2008vx}.

\vspace{12pt}

\bibliographystyle{JHEP}
\bibliography{report.bib}

\end{document}